\documentclass{article}
\usepackage{ijcai22}

\usepackage{times}

\usepackage{soul}
\usepackage{url}
\usepackage[hidelinks]{hyperref}
\usepackage[utf8]{inputenc}
\usepackage[small]{caption}
\usepackage{graphicx}
\usepackage{amsmath}
\usepackage{amsfonts}
\usepackage{amssymb}
\usepackage{amsthm}
\usepackage{booktabs}
\usepackage{stackrel}
\usepackage{algorithm}
\usepackage{algpseudocode}
\usepackage{subfigure}
\usepackage{epstopdf}
\usepackage{wrapfig}
\usepackage{tikz,pgfplots}
\usepackage{pgfplotstable}
\usepackage{balance}
\newcommand{\calL}{\mathcal{L}} 
\newcommand{\calB}{\mathcal{B}} 
\newcommand{\weight}{\omega}
\newcommand{\posI}{\mathcal{I}^{+}} 
\newcommand{\negI}{\mathcal{I}^{-}} 
\newcommand{\bfU}{\bar U}
\newcommand{\bfI}{\bar I} 
\newcommand{\bR}{\mathbb{R}} 

\newcommand{\prefu}{\stackrel[u]{}{\succ}}

\newcommand{\NetF}{\textsc{Netflix}}
\newcommand{\RecS}{\textsc{RecSys'16}}
\newcommand{\Out}{\textsc{Outbrain}}
\newcommand{\ML}{\textsc{ML}}
\newcommand{\kasandr}{\textsc{Kasandr}}
\newcommand{\PANDOR}{\textsc{Pandor}}

\newtheorem{theorem}{Theorem}

\newtheorem{assumption}{Assumption}

\newcommand{\SO}{\texttt{SAROS}}
\newcommand{\BPR}{\texttt{BPR}}
\newcommand{\MostPop}{\texttt{MostPop}}

\newcommand{\MF}{\texttt{MF}}
\newcommand{\GRU}{\texttt{GRU4Rec}}
\newcommand{\ProdVec}{\texttt{Prod2Vec}}

\newcommand{\batch}{\texttt{BPR$_b$}}
\newcommand{\caser}{\texttt{Caser}}
\newcommand{\GCN}{\texttt{LightGCN}}
\newcommand{\SASR}{\texttt{SASRec}}

\newcommand{\mapk}{\texttt{MAP@K}}

\newcommand{\mapfive}{\texttt{MAP@5}}
\newcommand{\mapten}{\texttt{MAP@10}}
\newcommand{\ndcgfive}{\texttt{NDCG@5}}
\newcommand{\ndcgten}{\texttt{NDCG@10}}
\newcommand{\ndcgk}{\texttt{NDCG@K}}

\title{Learning over No-Preferred and Preferred Sequence of Items for Robust Recommendation (Extended Abstract)}

\author{
Aleksandra Burashnikova$^1$\footnote{Contact Author}\and
Yury Maximov$^{2}$\and
Marianne Clausel$^{3}$\and
Charlotte Laclau$^4$\\
Franck Iutzeler$^5$\And Massih-Reza Amini$^5$\\
\affiliations
$^1$Skolkovo Institute of Science and Technology\\
$^2$ Los Alamos National Laboratory\\
$^3$ University of Lorraine\\
$^4$ Telecom Saint-Etienne \\
$^5$ University Grenoble Alpes\\
\emails
Aleksandra.Burashnikova@skoltech.ru,
yury@lanl.gov,
Marianne.Clausel@univ-lorraine.fr,
Charlotte.laclau@univ-st-etienne.fr,
franck.iutzeler@univ-grenoble-alpes.fr,
Massih-Reza.Amini@univ-grenoble-alpes.fr
}

\begin{document}

\maketitle

\begin{abstract}
This paper is an extended version of \cite{DBLP:journals/jair/BurashnikovaMCL21}, where we proposed a theoretically supported sequential strategy for training a large-scale Recommender System (RS) over implicit feedback, mainly in the form of clicks. The proposed approach consists in minimizing pairwise ranking loss over blocks of consecutive items constituted by a sequence of non-clicked items followed by a clicked one for each user. We present two variants of this strategy where model parameters are updated using either the momentum method or a gradient-based approach. To prevent updating the parameters for an abnormally high number of clicks over some targeted items (mainly due to bots), we introduce an upper and a lower threshold on the number of updates for each user. These thresholds are estimated over the distribution of the number of blocks in the training set. They affect the decision of RS by shifting the distribution of items that are shown to the users. Furthermore, we provide a convergence analysis of both algorithms and demonstrate their practical efficiency over six large-scale collections with respect to various ranking measures.
\end{abstract}

\section{Introduction}

The paper presents two variants of a sequential learning strategy for recommender systems with implicit feedback. The first approach, referred to as $\SO_m$, updates the model parameters at each time a block of unclicked items followed by a clicked one is formed after a user's interaction. Parameters' updates are carried out by minimizing the average ranking loss of the current model that scores the clicked item below the unclicked ones using a momentum method \cite{PolyakB}.
The second strategy, which we refer to as $\SO_b$, updates the model parameters by minimizing a ranking loss over the same blocks of unclicked items followed by a clicked one using a gradient descent approach; with the difference that parameter updates are discarded for users who interact very little or a lot with the system.

We present a unified framework in which we investigate the convergence qualities of both variants of \SO{} in the broad situation of non-convex ranking losses in this research.
The letter builds on our previous findings \cite{Burashnikova19}, which focused exclusively on the  convergence of  $\SO_b$ in the scenario of convex ranking losses. Furthermore, we provide empirical evaluation over six large publicly available datasets showing that both versions of $\SO_{}$ are highly competitive compared to the state-of-the-art models in terms of quality metrics. 

\section{Framework and Problem Setting}

\subsection{Learning Objective}
\label{sec:LO}
Our objective here is to minimize an expected error penalizing the misordering of all pairs of interacted items $i\in\posI_u$ and $i'\in\negI_u$ for a user $u$ where the set of preferred and non-preferred items denoted by $\posI_u$ and $\negI_u$, respectively. Commonly, this objective is given under the Empirical Risk Minimization (ERM) principle, by minimizing the empirical ranking loss estimated over the items and the final set of users who interacted with the system:
\begin{align}\label{eq:RL}
    \hat{\calL}_u(\weight)\!=\!\frac{1}{|\posI_u||\negI_u|}\!\sum_{i\in \posI_u}\!\sum_{i'\in \negI_u} \!\ell_{u,i,i'} (\weight),
\end{align}
\noindent where, $\ell_{u, i, i'} (.)$ is an instantaneous ranking loss defined over the triplet $(u,i,i')$ with the user $u$ prefers item $i$ over item $i'$; symbolized by the relation $i\!\prefu\! i'$. Hence,  $\hat{\cal L}_u(\weight)$ is the pairwise ranking loss with respect to user's interactions and ${\cal L}(\weight) = \mathbb{E}_u \left[\hat{\cal L}_u(\weight)\right]$ is the expected ranking loss, where $\mathbb{E}_{u}$ is the expectation with respect to users chosen randomly according to the marginal distribution. 

Each user $u$ and each item $i$ are represented respectively by vectors $\bfU_u\in\bR^k$ and $\bfI_i\in\bR^k$ in the same latent space of dimension $k$. The set of weights to be found $\weight=(\bfU,\bfI)$, are then matrices formed by the vector representations  of users $\bfU=(\bfU_u)_{u\in [N]}\in\bR^{N\times k}$ and items $\bfI=(\bfI_i)_{i\in[M]}\in\bR^{M\times k}$. 
The instantaneous loss, $\ell_{u,i,i'}$, is the surrogate regularized logistic  loss for some hyperparameter $\mu \ge 0$:
\begin{align}
\label{eq:instloss}
    \ell_{u,i,i'}(\weight) =  \log\left(1+e^{-y_{u,i,i'}\bfU_u^\top(\bfI_{i}-\bfI_{i'})}\right)+\\
    \mu (\|\bfU_u\|_2^2+\|\bfI_{i}\|_2^2 + \|\bfI_{i'}\|_2^2)\nonumber
\end{align}

In the case where one of the chosen items is preferred over the other one (i.e., $y_{u,i,i'}\in\{-1,+1\}$ and $y_{u,i,i'}=+1$ iff   $i\!\prefu\! i'$), the algorithm  updates the weights using the stochastic gradient descent method over the instantaneous loss \eqref{eq:instloss}.

\subsection{Algorithm \SO{}}
\label{sec:Algo}

 A key point in recommendation is that user preferences for items are largely determined by the context in which they are presented to the user. This effect of local preference is not taken into account by randomly sampling triplets formed by a user and corresponding clicked and unclicked items over the entire set of shown items. Furthermore, triplets corresponding to different users are non uniformly distributed, as interactions vary from one user to another one, and for parameter updates; triplets corresponding to low interactions have a small chance to be chosen. In order to tackle these points; we propose to update the parameters sequentially over the blocks of non-preferred items followed by preferred ones for each user $u$. 

In this case, at each time $t$ a block $\calB_u^t=\text{N}_u^{t}\sqcup\Pi_u^{t}$ is formed for user $u$; weights are updated by miniminzing the ranking loss corresponding to this block~:

\begin{equation}
\label{eq:CLoss}
    {\hat {\calL}}_{\calB_u^t}(\weight_u^{t}) = \frac{1}{|\Pi_u^{t}||\text{N}_u^{t}|}\sum_{i \in \Pi_u^{t}} \sum_{i'\in \text{N}_u^{t}} \ell_{u, i, i'} ({\weight}_u^{t}).
\end{equation}

We propose two strategies for the minimization of (Eq. \ref{eq:CLoss}) and the update of weights. In the first one, referred to as {\SO$_m$}, the aim is to carry out an effective minimization of the ranking loss \eqref{eq:CLoss} by lessening the oscillations of the updates through the minimum. This is done by defining the updates as the linear combination of the gradient of the loss of (Eq. \ref{eq:CLoss}), $\nabla  \widehat{\calL}_{{\mathcal B}^t_u}(w_u^{t})$, and the previous update as in the momentum technique at each iteration $t$~:

\begin{align}
\label{thm11:eq1}
    v_u^{t+1}&=\mu\cdot v_u^{t}+(1-\mu)\nabla  \widehat{\calL}_{{\mathcal B}^t_u}(w_u^{t})\\
    w_u^{t+1}&= w_u^{t}-\alpha v_u^{t+1}
\end{align}

\noindent 
where $\alpha$ and $\mu$ are hyperparameters of the linear combination. In order to explicitly take into account bot attacks -- in the form of excessive clicks over some target items -- we propose a second variant of this strategy, referred to as {$\SO_b$}. This variant consists in fixing two thresholds $b$ and $B$  over the parameter updates. For a new user $u$, model parameters are updated if and only if the number of  blocks of items constituted for this user is within the interval $[b,B]$. 

 The initial weights of the algorithms $\weight_1^0$ are chosen randomly for the first user. The sequential update rule of {$\SO_b$}, for each current user $u$ consists in updating the weights by making one step towards the opposite direction of the gradient of the ranking loss estimated on the current block,  $\calB_u^t=\text{N}_u^{t}\sqcup\Pi_u^{t}$~:

\begin{equation}
\weight_u^{t+1} = \weight_u^t - \frac{\eta}{|{\text{N}_u^{t}}||\Pi_u^{t}|}  \displaystyle{\sum_{i\in\Pi_u^{t}}\sum_{i'\in \text{N}_u^{t}}} \nabla \ell_{u,i,i'} (\weight_{u}^t)
\end{equation}

For a given user $u$, parameter updates are discarded if the number of blocks $(\calB_u^t)_t$ for the current user falls outside the interval $[b,B]$.
In this case, parameters are initialized with respect to the latest update before user $u$ and they are updated with respect to a new user's interactions.

\subsection{Convergence Analysis}
\label{sec:TA}

The proofs of algorithms' convergence are given under a common hypothesis that the sample distribution is not instantaneously affected by learning of the weights, i.e. the samples can be considered as i.i.d. More precisely, we assume the following hypothesis.

\begin{assumption}\label{ass:1}
    For an i.i.d. sequence of user and any $u, t \ge 1$, we have 
    \begin{enumerate}
        \item 
        $\mathbb{E}_{(u,{\cal B}_u^t)} \|\nabla {\calL}(\omega_u^t) - \nabla \hat{\calL}_{{\cal B}_u^t}(\omega_u^t)\|_2^2 \le \sigma^2$,
        \item  For any $u$,  $\left|\mathbb{E}_{ {{\cal B}_u^t}|u} \langle\nabla {\calL}(\omega_u^t), \nabla {\calL}(\omega_u^t) -  \nabla \hat{\calL}_{{\cal B}_u^t}(\omega_u^t) \rangle \right| \le a^2 \|\nabla {\calL}(\omega_u^t)\|_2^2$
    \end{enumerate}
    for some parameters $\sigma>0$ and $a\in [0,1/2)$ independent of $u$ and $t$. 
\end{assumption}

The first assumption is common in stochastic optimization and it implies consistency of the sample average approximation of the gradient. However, this assumption is not sufficient to prove the convergence because of interdependency of different blocks of items for the same user. 

The second assumption implies that in the neighborhood of the optimal point, we have $\nabla {\calL}(\omega_u^t)^\top \nabla \hat{\calL}_{{\cal B}_u^t}(\omega_u^t) \approx \|\nabla {\calL}(\omega_u^t)\|_2^2$, which greatly helps to establish consistency and convergence rates for both variants of the methods. 

The following theorem establishes the convergence rate for the \SO{}$_b$ algorithm. 

\begin{theorem}\label{thrm:new-01}
Let $\ell$ be a (possibly non-convex) $\beta$-smooth loss function. Assume, moreover,  that the number of interactions per user belongs to an interval $[b, B]$ almost surely and assumption~\ref{ass:1} is satisfied with some constants $\sigma^2$ and $a$, $0 < a < 1/2$. 
Then, for a step-size policy $\eta_u^t \equiv \eta_u$ with $\eta_u\leq 1/(B\beta)$ for any user $u$, one has
\begin{align}
\min_{1\le u \le N} & \mathbb{E}    \|\nabla {\calL}(\omega_u^0)\|_2^2 \le \nonumber\\ & \frac{2({\calL}(\omega_1^0) - {\calL}(\omega_u^0)) + \beta \sigma^2 \sum_{u=1}^N \sum_{t=1}^{|{\cal B}_u|}(\eta_u^t)^2}{\sum_{u=1}^N \sum_{t=1}^{|{\cal B}_u|} \eta_u^t(1 - a^2 - \beta \eta_u^t(1/2 - a^2))}
\end{align}
In particular, for a constant step-size policy $\eta_u^t = \eta = c/\sqrt{N}$ satisfies $\eta \beta \le 1$, one has 
\begin{align*}
    \min_{t, u} \|\nabla {\calL}(\omega_{u}^t)\|_2^2 \le 
    \frac{2}{b} \frac{2 ({\calL}(\omega_1^0) - {\calL}(\omega_*))/c + \beta c \sigma^2 B}{(1-4a^2)\sqrt{N}}. 
\end{align*}
\end{theorem}

\begin{proof}
Since $\ell$ is a $\beta$ smooth function, we have for any $u$ and $t$:
\begin{align*}
    {\calL}(\omega_{u}^{t+1}) & \le {\calL}(\omega_{u}^t) + \langle\nabla {\calL}(\omega_{u}^t), \omega_u^{t+1} - \omega_u^{t} \rangle \\
    & + \frac{\beta}{2}(\eta_u^t)^2 \|\nabla \hat{\calL}_{{\cal B}_u^t}(\omega_u^t)\|_2^2
    = {\calL}(\omega_{u}^t)\\
    & - \!\eta_u^t \langle \nabla {\calL}(\omega_{u}^t), \nabla \hat{\calL}_{{\cal B}_u^t}(\omega_u^t) \rangle \!+   \!\frac{\beta}{2}(\eta_u^t)^2 \|\nabla \hat{\calL}_{{\cal B}_u^t}(\omega_u^t)\|_2^2
\end{align*}    
Following~\cite{lan2020first}; by denoting $\delta_u^t = \nabla \hat{\calL}_{{\cal B}_u^t}(\omega_u^t) - \nabla {\calL}(\omega_u^t)$, we have: 
 \begin{align}\label{eq:01}
    {\calL}(\omega_{u}^{t+1})   & \le {\calL}(\omega_{u}^t) - \eta_u^t \langle \nabla {\calL}(\omega_{u}^t), \nabla {\calL}(\omega_{u}^t) + \delta_u^t \rangle \nonumber\\
    & + \frac{\beta}{2}(\eta_u^t)^2 \|\nabla {\calL}(\omega_u^t) + \delta_u^t\|_2^2= 
    {\calL}(\omega_{u}^t) \nonumber\\
    & +  \frac{\beta(\eta_u^t)^2}{2}\|\delta_u^t\|_2^2 - \!\!\left(\eta_u^i - \frac{\beta(\eta_u^t)^2}{2}\right)\!\!\|\nabla {\calL}(\omega_{u}^t)\|_2^2 \nonumber\\
    & - \left(\eta_u^t - \beta (\eta_u^t)^2\right) \langle\nabla {\calL}(\omega_u^t), \delta_u^t\rangle 
\end{align}

\noindent Our next step is to take the expectation on  both sides of inequality~\eqref{eq:01}. According to  Assumption~\ref{ass:1}, one has for some $a\in [0, 1/2)$:
\begin{align*}
    \left(\eta_u^t - \beta (\eta_u^t)^2\right)\left|\mathbb{E} \langle\nabla {\calL}(\omega_u^t), \delta_u^t\rangle\right| \le\\
    \left(\eta_u^t - \beta (\eta_u^t)^2\right) a^2 \|\nabla \calL (\omega_u^t)\|_2^2, 
\end{align*}
where the expectation is taken over the set of blocks and users seen so far. 

Finally, taking the same expectation on both sides of inequality~\eqref{eq:01}, it comes:
\begin{align}\label{eq:02}
    {\calL}(\omega_{u}^{t+1}) &\le  {\calL}(\omega_{u}^{t}) + \frac{\beta}{2}(\eta_u^t)^2\mathbb{E}\|\delta_u^t\|_2^2 - \nonumber\\
    &\eta_u^t(1 - \beta\eta_u^t/2  - a^2|1 - \beta\eta_u^t|) \|\nabla {\calL}(\omega_{u}^t)\|_2^2 \nonumber\\
    & \le  {\calL}(\omega_{u}^{t}) + \frac{\beta}{2}(\eta_u^t)^2\|\delta_u^t\|_2^2 \nonumber\\
    & - \eta_u^t\underbrace{(1 - a^2 
    -    \beta \eta_u^t(1/2 - a^2))}_{:= z_u^t} \|\nabla {\calL}(\omega_{u}^t)\|_2^2 \nonumber\\
    & = {\calL}(\omega_{u}^{t}) + \frac{\beta}{2}(\eta_u^t)^2\|\delta_u^t\|_2^2 - 
    \eta_u^t z_u^t \|\nabla {\calL}(\omega_{u}^t)\|_2^2 \nonumber\\
    & = {\calL}(\omega_{u}^{t}) + \frac{\beta}{2}(\eta_u^t)^2\sigma^2 - 
    \eta_u^t z_u^t \|\nabla {\calL}(\omega_{u}^t)\|_2^2, 
\end{align}
where the second inequality is due to  $|\eta_u^t\beta|\leq 1$. Also, as $|\eta_u^t\beta|\leq 1$ and $a^2\in [0,1/2)$ one has $z_u^t>0$ for any $u,t$. Rearranging the terms, one has
\begin{align*}
    \sum_{u=1}^N\!\sum_{t=1}^{|{\cal B}_u|} \!\eta_u^t z_u^t \|\nabla {\calL}(\omega_{u}^t)\|_2^2 \le \!{\calL}(\omega_1^0)\! - \!{\calL}(\omega_*) +   \sum_{u=1}^N\!\sum_{t=1}^{|{\cal B}_u|}
    \!\frac{\beta \sigma^2 (\eta_u^t)^2}{2}
\end{align*}
and
\begin{align*}
    \min_{t, u} \|\nabla {\calL}(\omega_{u}^t)\|_2^2
    & 
    \le 
    \frac{{\calL}(\omega_1^0) - {\calL}(\omega_*) + \frac{\beta}{2}  \sum_{u=1}^N\sum_{t=1}^{|{\cal B}_u|}(\eta_u^t)^2 \sigma^2 }{\sum_{u=1}^N\sum_{t=1}^{|{\cal B}_u|} \eta_u^t z_u^t} \\
    & 
    \le \frac{{\calL}(\omega_1^0) - {\calL}(\omega_*) + \frac{\beta}{2}  \sum_{u=1}^N\sum_{t=1}^{|{\cal B}_u|}(\eta_u^t)^2 \sigma^2 }{\sum_{u=1}^N\sum_{t=1}^{|{\cal B}_u|} \eta_u^t (1 - a^2 - \beta \eta_u^t(1/2 - a^2))} 
\end{align*}
Where, $\omega_*$ is the optimal point. Then, using a constant step-size policy, $\eta_u^i = \eta$, and the bounds on a block size, $b\leq |{\cal B}_u|\leq B$, we get:
\begin{align*}
    \min_{t, u} \|\nabla &{\calL}(\omega_{u}^t)\|_2^2  \le \frac{{\calL}(\omega_1^0) - {\calL}(\omega_*) + \frac{\beta\sigma^2}{2}  N\sum_{u=1}^N\eta_u^2  }{b\sum_{u=1}^N\eta_u (1 - a^2 - \beta \eta_u (1/2 - a^2))} \\
    &
    \le  \frac{4{\calL}(\omega_1^0) - 4{\calL}(\omega_*) + 2\beta \sigma^2 B \sum_{u=1}^N\eta^2}{b(1 - 4a^2)\sum_{u=1}^N\eta} \\
    &
    \le 
    \frac{2}{b(1-4a^2)} \left\{\frac{2{\calL}(\omega_1^0) - 2{\calL}(\omega_*)}{N \eta} + \beta\sigma^2 B \eta\right\}. 
\end{align*}

Taking $\eta = c/\sqrt{N}$ so that $0 < \eta \le 1/\beta$, one has 
\begin{align*}
    \min_{t, u} \|\nabla {\calL}(\omega_{u}^t)\|_2^2 \le 
    \frac{2}{b} \frac{2 ({\calL}(\omega_1^0) - {\calL}(\omega_*))/c + \beta c \sigma^2 B}{(1-4a^2)\sqrt{N}}. 
\end{align*}
If $b = B = 1$, this rate matches up to a constant factor to the standard $O(1/\sqrt{N})$ rate of the stochastic gradient descent. 
\end{proof}

The analysis of momentum algorithm \SO{}$_m$
is slightly more involved and based on the Polyak-\L{}ojsievich condition \cite{Polyak63,karimi2016linear}. 
Based on the latter condition we provide an analysis on the convergence of \SO{}$_m$ in  \cite{DBLP:journals/jair/BurashnikovaMCL21}. Also, we notice that this strategy can be useful in analysis of multi-class classification problems \cite{joshi2017aggressive,maximov2018rademacher,maximov2018rademacher2} and complements earlier results on ranking algorithms convergence~\cite{moura2018heterogeneous,sidana2021user}.

\section{Experimental Setup and Results}
\label{sec:Exps}
\begin{table*}[hbt!]
    \centering
     \resizebox{.75\textwidth}{!}{\begin{tabular}{c|cccccc|cccccc}
    \hline
     &\multicolumn{6}{c|}{\ndcgfive}&\multicolumn{6}{c}{\ndcgten}\\
     \cline{2-13}
     &{\ML-1M}&\Out&\PANDOR&{\NetF}&{\kasandr}&{\RecS}&{\ML-1M}&\Out&\PANDOR&{\NetF}&{\kasandr}&{\RecS}\\
     \hline
     \MostPop  & .090&.011 &.005 &.056 & .002&.004& .130&.014 &.008 &.096 &.002&.007\\
     \ProdVec  & .758& .232& .078& .712& .012&.219& .842& .232&.080 &.770 &.012&.307 \\
     \MF       & .684&.612 &.300 & .795& .197&.317& .805& .684& .303& .834& .219&.396\\
    \batch     & .652& .583& .874& .770& .567&.353& .784& .658& .890& .849& .616&.468\\
     \BPR      & .776& \underline{.671}& .889& \underline{.854}& .603 &\bf{.575}&\underline{.863}& \underline{.724}& .905& .903& .650&\bf{.673}\\
     {\GRU}$+$ & .721& .633 & .843& .777& .760&.507& .833& .680&.862 &.854 &.782&.613\\
      {\caser} & .665& .585& .647& .750& .241&.225& .787& .658& .666& .834& .276&.225\\
     {\SASR}   & .721 & .645& .852& .819& .569&.509& .832& .704& .873& .883& .625&.605\\
     {\GCN} & \underline{.784} & .652 & \underline{.901} & .836 & \bf{.947} & .428 & \bf{.874} & .710 & \underline{.915} & .895 & \bf{.954} & .535\\
     {$\SO_m$}&.763&.674&.885&.857&.735&.492&.858&.726&.899&\underline{.909}&.765&.603\\
     {$\SO_b$}     & \bf{.788}& \bf{.710}& \bf{.904}& \bf{.866}& \underline{.791} &\underline{.563}& \bf{.874}& \bf{.755}& \bf{.917}& \bf{.914}& \underline{.815}&\underline{.662}\\

     \hline
    \end{tabular}
    }
    ~\\~\\     \resizebox{0.75\textwidth}{!}{\begin{tabular}{c|cccccc|cccccc}
    \hline
      &\multicolumn{6}{c|}{\mapfive}&\multicolumn{6}{c}{\mapten}\\
     \cline{2-13}
     &{\ML-1M}&\Out&\PANDOR&{\NetF}&{\kasandr}&{\RecS}&{\ML-1M}&\Out&\PANDOR&{\NetF}&{\kasandr}&{\RecS}\\
     \hline
     \MostPop  & .074&.007 &.003 &.039 & .002&.003& .083&.009 &.004 &.051 &.3e-5&.004 \\
     \ProdVec  & .793& .228& .063& .669& .012&.210& .772& .228&.063 &.690 &.012&.220 \\
     \MF       & .733&.531 &.266 & .793& .170&.312& .718& .522& .267& .778& .176&.306\\
    \batch     & .713& .477& .685& .764& .473&.343& .688& .477& .690& .748& .488&.356\\
     \BPR      & 
     .826& .573& .734& .855& .507&\bf{.578}&.797& .563& \underline{.760} & \underline{.835}& .521&\bf{.571}\\
     {\GRU}$+$ & .777& .513 & .673& .774& .719& .521& .750& .509&.677 &.757 &\underline{.720}&.500\\
      {\caser} & .718& .471& .522& .749& .186&.218& .694& .473& .527& .733& .197&.218\\
     {\SASR}   & .776 & .542& .682& .819& .480&.521& .751& .534& .687& .799& .495&.511\\
    {\GCN} & \bf{.836} & .502 & \bf{.793} & .835 & \bf{.939} & .428 & \underline{.806} & .507 & \bf{.796} & .817 & \bf{.939} & .434\\
     $\SO_m$&.816&\underline{.577}&.720&\underline{.857}&.644&.495&.787&\underline{.567}&.723&.837&.651&.494\\
     $\SO_b$     & \underline{.832}& \bf{.619}& \underline{.756} & \bf{.866}& \underline{.732}&\underline{.570}& \bf{.808}& \bf{.607}& .759& \bf{.846}& .747&\underline{.561}\\
     \hline
    \end{tabular}
    }
    \caption{Comparison between \MostPop, \ProdVec, \MF, \batch, \BPR{}, {\GRU}$+$, \SASR, \caser, and \SO{} approaches in terms of \ndcgfive{} and \ndcgten (top), and \mapfive{} and \mapten (down). Best performance is in bold and the second best is underlined.}
    \label{tab:online_vs_all_ndcg_1h}
\end{table*}

\paragraph{Datasets.} We report results obtained on six publicly available data\-sets, for the task of personalized Top-N recommendation on the following collections: \ML-1M~\cite{Harper:2015:MDH:2866565.2827872}, {\NetF}~\cite{Bennett07}, {\PANDOR}~\cite{pandor}, {\RecS}~ that is a sample based on historic XING data, {\kasandr}~\cite{sidana17} and a subset out of the {\Out}~ dataset from of the Kaggle challenge\footnote{\url{https://www.kaggle.com/c/outbrain-click-prediction}}.

\paragraph{Compared Approaches. } To validate the sequential learning approach described in the previous sections, we compared the proposed \SO{} algorithm\footnote{The source code is available at \url{https://github.com/SashaBurashnikova/SAROS}.} with the following approaches.  
\begin{itemize}
    \item {\MostPop} is a non-learning based approach which consists in recommending the same set of popular items.
    \item Matrix Factorization ({\MF}) \cite{Koren08}, decomposes the matrix of user-item interactions, by minimizing a regularized least square error between the actual value of the scores and the dot product over representations. 
    \item {\BPR} \cite{rendle_09} a stochastic  gradient-descent  algorithm, based  on  bootstrap  sampling  of  training  triplets, and {\batch} the batch version of the model.
   \item {\ProdVec} \cite{GrbovicRDBSBS15},   performs next-items recommendation based on the similarity between the representations of items using word2vec.
   
    \item {\GRU}$+$ \cite{hidasi2018recurrent}, learns model parameters by optimizing a regularized approximation of the relative rank of the relevant item which favors top ranked preferred items.
    \item {\caser} \cite{tang2018caser} embeds a sequence of clicked items into a temporal image and latent spaces and find local characteristics of the temporal image using convolution filters. 
    
    \item {\SASR} \cite{DBLP:conf/icdm/Kang18} uses  an attention mechanism to capture long-term semantics in the sequence of clicked items.

     \item {\GCN} \cite{He2020} is a graph convolution network which learns user and item embedding by linearly propagating them on the user-item interaction graph.

\end{itemize}

\paragraph{Evaluation Setting and Results. }




We compare the performance of all the approaches on the basis of the common ranking metrics, which are the Mean Average Precision at rank $K$ (\mapk) and the Normalized Discounted Cumulative Gain at rank $K$ (\ndcgk) that computes the ratio of the obtained ranking to the ideal case and allow to consider not only binary relevance as in Mean Average Precision.

Table \ref{tab:online_vs_all_ndcg_1h} presents {\ndcgfive} and {\ndcgten} (top), and {\mapfive} and {\mapten} (down)  of  all approaches over the test sets of the different collections.
The non-machine learning method, {\MostPop},  gives results of an order of magnitude lower than the learning based approaches. Moreover,  the factorization model {\MF} which predicts clicks by matrix completion is less effective when dealing with implicit feedback than ranking based models. We also found that embeddings of ranking based models are more robust than the ones found by {\ProdVec}. When comparing {\GRU}$+$ with {\BPR} that also minimizes the same surrogate ranking loss, the former outperforms it in case of {\kasandr} with a huge imbalance between positive and negative interactions. This is mainly because  {\GRU}$+$ optimizes an approximation of the relative rank that favors interacted items to be in the top of the ranked list while the logistic ranking loss, which is mostly related to the Area under the ROC curve \cite{Usunier:1121}, pushes up clicked items for having good ranks in average. However, the minimization of the logistic ranking loss over blocks of very small size pushes the clicked item to be ranked higher than the no-clicked ones in several lists of small size and it has the effect of favoring the clicked item to be at the top of the whole merged lists of items. Moreover,  it comes out that {\SO} is the most competitive approach; performing better than other techniques, or, is the second best performing method.

\section{Conclusion}\label{sec:Conclusion}

The contributions of this paper are twofold. First, we proposed {\SO},  a novel learning framework for large-scale Recommender Systems that sequentially updates the weights of a ranking function user by user over blocks of items ordered by time where each block is a sequence of negative items followed by a last positive one. The main hypothesis of the approach is that the preferred and no-preferred items within a local sequence of user interactions express better the user preference than when considering the whole set of preferred and no-preferred items independently one from another. The second contribution is a theoretical analysis of the proposed approach which bounds the deviation of the ranking loss concerning the sequence of weights found by both variants of the algorithm and its minimum in the general case of non-convex ranking loss.  Empirical results conducted on six real-life implicit feedback datasets support our founding and show that the proposed approach is highly competitive concerning state of the art approaches on \texttt{MAP} and \texttt{NDCG} measures.

\section*{Acknowledgements}
{AB is supported by the Analytical center under the RF Government (subsidy agreement 000000D730321P5Q0002, Grant No. 70-2021-00145 02.11.2021). YM is supported by LANL LDRD projects.
}
\balance
\bibliographystyle{named}
\bibliography{ijcai22}

\end{document}